\documentclass[letterpaper, 10 pt, conference]{ieeeconf}  

\IEEEoverridecommandlockouts                              

\usepackage[left=0.75in,right=0.75in,top=1in,bottom=0.75in]{geometry}

\usepackage{amssymb}
\usepackage{amsmath}
\usepackage{graphicx}
\usepackage{bm}
\usepackage{url}

\usepackage[utf8]{inputenc}
\usepackage{pgfplots}
\usepgfplotslibrary{groupplots,dateplot}
\usetikzlibrary{patterns,shapes.arrows}
\pgfplotsset{compat=newest}

\newtheorem{definition}{Definition}

\newtheorem{proposition}{Proposition}

\newtheorem{remark}{Remark}

\usepackage{tikz}
\usepackage{tikz-3dplot}
\usetikzlibrary{calc}

\usepackage{algorithm}
\usepackage{algpseudocode}
\usepackage{soul} 
\usepackage{siunitx}
\usepackage{subcaption} 
\usepackage{hyperref}

\usepackage[autostyle]{csquotes}
\MakeAutoQuote{‘}{’}

\def\Inf{\operatornamewithlimits{inf\vphantom{p}}}

\usepackage[backend=bibtex, doi=false,isbn=false,url=false,eprint=false, autolang=none, sorting=none]{biblatex}
\graphicspath{/Figures/}

\definecolor{tum blue}{HTML}{0065BD}
\definecolor{tum blue 1}{HTML}{98C6EA}
\definecolor{tum blue 2}{HTML}{64A0C8}
\definecolor{tum blue 3}{HTML}{0073CF}
\definecolor{tum blue 4}{HTML}{005293}
\definecolor{tum blue 5}{HTML}{003359}

\definecolor{tum green}{HTML}{A2AD00}
\definecolor{tum orange}{HTML}{E37222}
\definecolor{tum ivory}{HTML}{DAD7CB}

\definecolor{tum dia violet}{HTML}{69085A}
\definecolor{tum dia dark blue}{HTML}{0F1B5F}
\definecolor{tum dia turquoise}{HTML}{00778A}
\definecolor{tum dia dark green}{HTML}{007C30}
\definecolor{tum dia light green}{HTML}{679A1D}
\definecolor{tum dia light yellow}{HTML}{FFDC00}
\definecolor{tum dia dark yellow}{HTML}{F9BA00}
\definecolor{tum dia dark orange}{HTML}{D64C13}
\definecolor{tum dia red}{HTML}{C4071B}
\definecolor{tum dia dark red}{HTML}{9C0D16}

\newcommand{\R}{\mathbb{R}}

\newcommand{\trans}[1]{#1^\top}

\newcommand{\inv}[1]{#1^{-1}}

\newcommand{\diff}[2]{\frac{\mathrm{d} #1}{\mathrm{d} #2}}
\newcommand{\pdiff}[2]{\frac{\partial #1}{\partial #2}}

\newcommand\copyrighttext{%
	\footnotesize \textcopyright 2025 EUCA. Accepted at 23rd European Control Conference (ECC). Link to proceedings: \url{https://ieeexplore.ieee.org/abstract/document/11186824}.
}  
\newcommand\copyrightnotice{%
	\begin{tikzpicture}[remember picture,overlay]  
		\node[anchor=south,yshift=10pt, xshift=0pt] at (current page.south) {\fbox{\parbox{\dimexpr\textwidth-\fboxsep-\fboxrule\relax}{\copyrighttext}}};  
	\end{tikzpicture}%
	\vspace{-0.35cm}  
}

\title{\LARGE \bf
Control Lyapunov Functions for Optimality in Sontag-Type Control
}

\author{Joscha F. Bongard$^{1}$, Boris Lohmann$^{1}$
\thanks{$^{1}$Chair of Automatic Control, Department of Engineering Physics and Computation, Technical University of Munich,
        Boltzmannstraße 15, 85748 Garching bei München, Germany
        {\tt\small \{joscha.bongard,lohmann\}@tum.de}}%
\thanks{*funded by the Deutsche Forschungsgemeinschaft (DFG, German Research Foundation) - 469341384}
}

\begin{document}

\maketitle
\thispagestyle{empty}
\pagestyle{empty}
\copyrightnotice{}

\begin{abstract}

Given a Control Lyapunov Function (CLF), Sontag's famous Formula provides a nonlinear state-feedback guaranteeing asymptotic stability of the setpoint. At the same time, a cost function that depends on the CLF is minimized. While there exist methods to construct CLFs for certain classes of systems, the impact on the resulting performance is unclear. This article aims to make two contributions to this problem: (1) We show that using the value function of an LQR design as CLF, the resulting Sontag-type controller minimizes a classical quadratic cost around the setpoint and a CLF-dependent cost within the domain where the CLF condition holds. We also show that the closed-loop system is stable within a local region at least as large as that generated by the LQR. (2) We show a related CLF design for feedback-linearizable systems resulting in a \emph{global} CLF in a straight-forward manner; The Sontag design then guarantees global asymptotic stability while minimizing a quadratic cost at the setpoint and a CLF-dependent cost in the whole state-space. Both designs are constructive and easily applicable to nonlinear multi-input systems under mild assumptions.

\end{abstract}


\section{Preliminaries}

\subsection{Introduction}

Even though control theory made great advancements in the past decades, the goals of controlling nonlinear systems with a large region of attraction (ROA) and in an optimal fashion are often competing. While numerical methods like model predictive control \cite{Primbs1998CLFBasedRHC} aim to achieve both, they introduce challenges such as delays due to solving optimization problems.
Some well known nonlinear methods are feedback linearization, backstepping and forwarding and sliding mode control \cite{Khalil2002NonlinearSystems}. Many of these nonlinear design methods however suffer from strong requirements on the system or relying on nontrivial choices the designer has to make. This often limits nonlinear methods to simple systems with a small number of states \cite{Sepulchre1997}.
For smooth systems, the ubiquitous linear quadratic regulator (LQR) achieves quadratic optimal control locally, but the linearization fundamentally limits the method, often generating smaller ROA around the setpoint as well as performance which might quickly deteriorate away from the vicinity of the setpoint.
One promising way to achieve larger ROAs for nonlinear systems is Sontag's Formula \cite{Sontag1989UniversalConstructionOfArtsteinsTheorem}, which is a method of using a system's Control Lyapunov function (CLF) to achieve Lyapunov stability by design. It also promises optimality w.r.t. some cost functional \cite{Sepulchre1997,Lohmann2023DiscussionOnNonlinQCAndSontagsFormula}.\\
While this makes Sontag's Formula an appealing choice, the choice of the used CLF remains challenging, especially since it influences the cost functional the controller is optimal w.r.t. in an unclear way.

\subsection{Related Work}

Some of the most straightforward approaches in the stabilization of nonlinear systems are Lyapunov-based ones which leverage Lyapunov theory to reach stability by design. One of the most well-known practical examples of this family is Sontag's Formula. 

The control law used in this paper is a variant of Sontag's Formula, which was proposed in \cite{Primbs1998CLFBasedRHC}. Optimality of this controller w.r.t. a meaningful cost was shown in \cite{Sackmann2000ModifiedOptimalControl}, the cost was however not explicitly given. This cost can be found in \cite{Lohmann2023DiscussionOnNonlinQCAndSontagsFormula}, which turns out to be a quadratic cost distorted by a scalar state-dependent factor cf. \eqref{eq:inverse_optimal_cost}.
The effect of the choice of CLF on the inverse optimal cost is unclear and, to the best of our knowledge, has not been investigated until now.

Finding a CLF for a nonlinear system is generally difficult, as the set of CLFs may be non-convex or even non-connected.
There exists a multitude of different methods for constructing global CLFs. However, they are limited to specific classes of systems and generally don't give guarantees on the resulting performance \cite{Primbs1999Diss, Sepulchre1997, Freeman1996RobustNonlinearControlDesign}.

The use of quadratic CLFs is a common approach and can be found in many works on the topic, e.g. \cite{Sepulchre1997,Yu2001ComparisonOfNonlinearControl}.
The idea of approximating the value function locally by a Riccati-based CLF to obtain a locally optimal controller can be found in \cite{Primbs1999Diss}. The author does however not investigate this in detail.
In \cite{Yu2001ComparisonOfNonlinearControl}, the authors present multiple constructions for CLFs based on Riccati methods.

The idea of constructing a nonlinear state feedback giving a large ROA, which also locally aligns with the LQR for optimality was discussed more recently in \cite{Manchester2017} in the context of control contraction metrics. This contraction-based design is of comparable complexity to the CLF-based one presented here, requiring a control contraction metric instead of a CLF.

\subsection{Contribution}
The contribution of this paper is twofold: First, we show that using a Sontag-type formula with a local CLF based on LQR design, the resulting controller is locally optimal in the sense that it recovers the LQR for smooth nonlinear systems under mild assumptions. We show that this controller generates a Region of Attraction (ROA) with a provable subset based on linearization at least as large as the one provable for the LQR.\\
Second, we present a related method for constructing CLFs for feedback linearizable systems that also achieves local optimality of the resulting controller and additionally achieves global asymptotic stability.

All results are applicable to nonlinear input-affine multi-input systems under mild assumptions. 
Numerical experiments show the behavior of the proposed design compared to widespread designs also yielding local optimality such as the LQR or a properly adjusted feedback linearizing controller. The code used to generate all plots is publicly accessible (see Section \ref{sec:numerical_results}).

\subsection{Notation}

Denote the real numbers by $\R$, the non-negative real numbers by $\R_{\geq 0}$, the zero matrix of appropriate dimensions as $\bm{0}$. The Lie Derivative of a function $V(x)$ along a vector field $f(x)$ is denoted by $L_f V(x) := \pdiff{V}{x}(x) f(x)$. $\mathcal{C}^1$ and $\mathcal{C}^\infty$ denote the classes of functions that are continuously differentiable once and infinitely many times respectively. 
Let $R_{n}(x)$ denote a (possibly infinite) sum of multivariate monomials of at least order $n$. For example, the remainder terms of a truncated multivariate Taylor series expansion of order $n$ are denoted as $R_{n+1}(x)$.

\section{Control Lyapunov Function and Sontag-Type Control}

Consider the input affine system
\begin{align} \label{eq:input_affine_system}
	\dot{x} = f(x) + G(x)u
\end{align}
with $x \in \R^{n}$, $u \in \R^{m}$ and $G(x) \in \R^{n \times m}$, and $f(\bm{0}) = \bm{0}$. Assume that $f$ and $G$ are $\mathcal{C}^\infty$. Denote $\pdiff{f}{x}|_{x=\bm{0}} = A$, $G(\bm{0}) = B$.
In order to derive stability results for this system, we consider the following two definitions of a CLF:
\begin{definition}[CLF \cite{Freeman1996RobustNonlinearControlDesign, Sepulchre1997}] \label{def:CLF_definition}
	A Control Lyapunov Function (CLF) $V$ for system \eqref{eq:input_affine_system} is a $\mathcal{C}^1$, positive definite, radially unbounded function: $V: \R^n \rightarrow \R$ such that
	\begin{align} \label{eq:CLF_condition}
	\Inf_{u} \left[ L_f V(x) + L_G V(x) u \right] < 0,
	\end{align}
	for all $x \neq \bm{0}$.
\end{definition}

\begin{definition}[local CLF] \label{def:local_CLF_definition}
	A local CLF $V$ is a function which satisfies the properties of Definition \ref{def:CLF_definition}, but the CLF inequality \eqref{eq:CLF_condition} holds in an open region $\mathcal{D} \subset \R^n$ containing $\bm{0}$.
\end{definition}



Following the original notation by E. D. Sontag, define $a(x) = L_f V(x)$, and $b(x) = L_G V(x)$, and 
consider the following \emph{Sontag-type control law} \cite{Primbs1998CLFBasedRHC,Sackmann2000ModifiedOptimalControl}
\begin{align} \label{eq:sontag_sackmann}
u_S(x) =
\begin{cases}
- \inv{R} \trans{b(x)} \lambda(x) & b(x) \neq \bm{0}\\
\bm{0} & b(x) = \bm{0}
\end{cases}
\end{align}
with the scalar factor $\lambda(x)$ given by
\begin{align} \label{eq:lambda}
	\lambda(x) = \frac{a(x) + \sqrt{ a^2(x) + \trans{x} Q x \ b\left(x\right) \inv{R} \trans{b}\left(x\right) }}{b\left(x\right) \inv{R} \trans{b}\left(x\right)},
\end{align}
where $Q, R \succ 0$ are the weighting matrices. This control law remains continuous in $x$ within the domain where \eqref{eq:CLF_condition} holds \cite{Sackmann2000ModifiedOptimalControl}. It achieves
\begin{align} \label{eq:lyap_decay_sontag}
	\diff{V}{t} = \begin{cases}
		- \sqrt{a^2(x)+  \trans{x} Q x b\left(x\right) \inv{R} \trans{b}\left(x\right) }, & b(x) \neq \bm{0}\\
		a(x) & b(x) = \bm{0}
	\end{cases}
\end{align}
which is negative definite, given that $V$ satisfies the CLF property \eqref{eq:CLF_condition}.
The control law \eqref{eq:sontag_sackmann} is of particular interest because it can be shown to be optimal w.r.t. the following cost functional \cite{Sackmann2000ModifiedOptimalControl, Sepulchre1997, Lohmann2023DiscussionOnNonlinQCAndSontagsFormula}
\begin{align} \label{eq:inverse_optimal_cost}
	J\left(x\left(t_0\right), u\left(\cdot\right)\right) = \frac{1}{2} \int_{t_0}^{\infty} \frac{1}{\lambda(x)} \left( \trans{x} Q x + \trans{u} R u \right) d \tau,
\end{align}
referred to as an \emph{inverse optimal} property, as shown in \cite{Freeman1996CLFNewIdeasFromAnOldSource,Lohmann2023DiscussionOnNonlinQCAndSontagsFormula}.
The cost \eqref{eq:inverse_optimal_cost} is similar to the \emph{standard quadratic cost}
\begin{align} \label{eq:standard_optimal_cost}
	J\left(x\left(t_0\right), u\left(\cdot\right)\right) = \frac{1}{2} \int_{t_0}^{\infty} \left( \trans{x} Q x + \trans{u} R u \right) d \tau.
\end{align}
but is ‘distorted’ by the state-dependent factor $\frac{1}{\lambda(x)}$ \cite{Lohmann2023DiscussionOnNonlinQCAndSontagsFormula}. If $V$ satisfies the CLF condition \eqref{eq:CLF_condition}, then $\lambda(x) > 0$ for $x \neq \bm{0}$ \cite{Sackmann2000ModifiedOptimalControl}.

The following sections aim at finding CLFs in such a way that $\lambda(x)$ becomes $1$, and thus standard quadratic optimality is recovered, at least locally. 

\section{Local CLF from LQR Design}


Assume the linearized system $\dot{x} = Ax + Bu$ with $(A,B)$ stabilizable and that we have designed an \emph{LQR feedback} (see \cite{Anderson1989}) for it, minimizing the standard cost \eqref{eq:standard_optimal_cost} for chosen weighting matrices $Q, R \succ 0$. The corresponding value function is
\begin{align}
	\label{eq:CLFCandidateFromLQR}
	V\left(x\right) = \frac{1}{2} x^\top P x,
\end{align}
where $P$ is the unique positive definite solution of the corresponding algebraic Riccati equation:
\begin{align} \label{eq:ARE}
	\trans{A}P + PA - PB\inv{R}\trans{B}P+Q = \bm{0}.
\end{align}
We propose to use this value function $V$ from \eqref{eq:CLFCandidateFromLQR} as a candidate CLF for the nonlinear control design \eqref{eq:sontag_sackmann}. Since $P$ is positive definite the LQR locally stabilizes the origin, $V$ is at least guaranteed to be a local CLF according to Definition \ref{def:local_CLF_definition}.

Let us now show that the proposed constructive design is locally optimal in the sense that it recovers the LQR feedback
\begin{align}
	u(x) &= -\inv{R}\trans{B}P x \label{eq:LQR}
\end{align}
at the equilibrium. This is done with the following novel proposition.
\begin{proposition} \label{prop:local_optimality}
	For an LTI system $\dot{x} = Ax + Bu$, the Sontag-type control law \eqref{eq:sontag_sackmann} reduces to the LQR feedback $u = - \inv{R} \trans{B} P x$ if the CLF $V(x) = \frac{1}{2} \trans{x} P x$ is used, where $P$ is the unique positive definite solution to the corresponding Algebraic Riccati equation \eqref{eq:ARE} for $Q, R \succ 0$.
\end{proposition}

\begin{proof}
	Substituting the linear system $f(x) = Ax$, $g(x) = B$ and the CLF choice \eqref{eq:CLFCandidateFromLQR}, the expressions in feedback \eqref{eq:sontag_sackmann} become
	\begin{subequations} \label{eq:a_b_expressions}
	\begin{align}
		a = \pdiff{V}{x} f &= \trans{x} PA x,\\
		b = \pdiff{V}{x} G &= \trans{x} PB.
	\end{align}
	\end{subequations}
	
	Using these expressions, the controller \eqref{eq:sontag_sackmann} becomes
	\begin{align}
	u(x) = - \lambda(x) \inv{R}\trans{B} P x.
	\end{align}
	This controller is equal to the LQR feedback multiplied by a state-dependent scalar factor $\lambda(x)$ as in \eqref{eq:lambda}. In the LTI case with the given CLF $V$, this factor becomes
	\begin{align} \label{eq:lambda_lti}
	\lambda(x) = \frac{\trans{x}PAx + \sqrt{\left(\trans{x}PAx\right)^2 + \trans{x}Qx \trans{x}PB\inv{R}\trans{B}Px }}{\trans{x}PB\inv{R}\trans{B}Px}.
	\end{align}
	Recall that the symmetric positive definite matrix $P$ satisfies the Algebraic Riccati equation (ARE) \eqref{eq:ARE} \cite{Anderson1989}.
	Multiplying the ARE \eqref{eq:ARE} by $(-1)$, then $\trans{x}$ from the left and $x$ from the right yields
	\begin{align} \label{eq:quadratic_form_xTPAx}
	\trans{x}PAx = \frac{1}{2} \trans{x}PB\inv{R}\trans{B}Px - \frac{1}{2}\trans{x}Qx.
	\end{align}
	Substituting \eqref{eq:quadratic_form_xTPAx} into \eqref{eq:lambda_lti} gives
	\begin{align} \label{eq:linear_optimality}
	\begin{split}
	\lambda(x) &= \frac{\trans{x}PAx + \sqrt{\left(\frac{1}{2}\trans{x}PB\inv{R}\trans{B}Px + \frac{1}{2} \trans{x}Qx\right)^2}}{\trans{x}PB\inv{R}\trans{B}Px}\\
	&\overset{\eqref{eq:quadratic_form_xTPAx}}{=} 1
	\end{split}
	\end{align}
	Since the system is stabilizable, the CLF condition holds, and hence the denominator in \eqref{eq:linear_optimality} is always positive (and thus nonzero) \cite{Sackmann2000ModifiedOptimalControl}.\\
	\eqref{eq:linear_optimality} thus shows that the LQR is recovered by the proposed controller at the point of linearization.
\end{proof}
The quadratic choice of CLF \eqref{eq:CLFCandidateFromLQR} is particularly appealing because it is constructive and it allows locally optimal nonlinear feedback as in Proposition \ref{prop:local_optimality} through the proposed controller \eqref{eq:sontag_sackmann} by approximating the quadratic part of the value function. It is also easy to adjust by tuning the weighting matrices $Q$ and $R$.

\begin{remark}
	This result complements discussions on local optimality of Sontag-type formulas e.g. in \cite{Primbs1999Diss} and to the best of our knowledge, is the first proof of this conjecture in \cite{Primbs1999Diss}.
\end{remark}

\begin{remark}
	If the CLF is of the form $V(x) = \frac{1}{2} x^\top P x + R_3(x)$, the result still holds locally, showing that aligning the quadratic part of the CLF to the value function is crucial for local performance.
\end{remark}

\begin{remark}
	A note on relations to the nonlinear feedback minimizing the standard cost \eqref{eq:standard_optimal_cost} by solving the Hamilton-Jacobi-Bellman equation can be found in the Appendix.
\end{remark}

Note that the CLF-dependent cost \eqref{eq:inverse_optimal_cost} is minimized throughout the whole domain where the CLF condition \eqref{eq:CLF_condition} holds \cite{Lohmann2023DiscussionOnNonlinQCAndSontagsFormula}.



Let us now have a look at the ROA. Since both the LQR and the Sontag-type approach with a local CLF only offer \emph{local} stability guarantees, it might be insightful to compare the ROA they generate around the origin.
In general, there is no guarantee on the size of the resulting ROA when applying linear controllers to smooth nonlinear systems \cite{Khalil2002NonlinearSystems}.
The following (novel) proposition shows that at least for a local subset of the ROA generated by the LQR, this set is also a subset of the ROA generated by the Sontag-type controller. The proof uses a Taylor series expansion and thus the discussed regions might be arbitrarily small.

\begin{proposition} \label{prop:ROA}
	Assume $(A,B)$ stabilizable.
	There exists a compact sublevel set of $V$ as in \eqref{eq:CLFCandidateFromLQR} and some $C \in \R_{\geq 0}$ s.t. the following set is non-empty:
	\begin{align} \label{eq:X_ROA_LQR}
		\mathbb{X}_{\mathrm{ROA,L}} := \left\{ x \in \mathbb{X} | x \neq \bm{0}, \; V(x) \leq C, \; \dot{V}(x) < 0 \right\},
	\end{align}
	where $\dot{V}$ is calculated under the LQR \eqref{eq:LQR}. $\mathbb{X}_{\mathrm{ROA,S}}$ denotes the analogous set with $\dot{V}$ under the Sontag-type controller.
	Then it follows that
	$\mathbb{X}_{\mathrm{ROA,L}} \subset \mathbb{X}_{\mathrm{ROA,S}}$.
\end{proposition}

\begin{proof}
	We first show that a nonempty state space region $\mathbb{X}_{\mathrm{ROA,L}}$ as in \eqref{eq:X_ROA_LQR} exists. It is then straightforward to show that the Sontag-type controller also leads to a Lyapunov decrease in the same region.
	Denoting $K = \inv{R}\trans{B}P$, we take the Taylor series expansion of the time derivative of the CLF candidate \eqref{eq:CLFCandidateFromLQR} under the LQR control law $u = -Kx$.
	\begin{align}
		\begin{split}
			\dot{V} &= \pdiff{V}{x} \left( \left( \pdiff{f}{x}|_{x=\bm{0}}-G\left(\bm{0}\right)K \right)x + R_{2}(x) \right)\\ 
			&= \trans{x}P \left(A-BK\right) x + R_{3}(x).
		\end{split}
	\end{align}
	The LQR for the stabilizable linearized system $(A,B)$ ensures that the symmetric part of the matrix $P \left(A-BK\right)$ is negative definite \cite{Anderson1989}, hence the origin is locally stable even under the original nonlinear system.
	Since $V$ is a quadratic form, there exists a region around $\bm{0}$ which is a subset of the ROA \cite{Khalil2002NonlinearSystems}.
	This shows that there exists a nonempty state space region $\mathbb{X}_{\mathrm{ROA,L}}$ \eqref{eq:X_ROA_LQR}, where $V$ satisfies the CLF condition \eqref{eq:CLF_condition}. Therefore, the Sontag-type controller \eqref{eq:sontag_sackmann} also leads to a Lyapunov decay in the same region.
\end{proof}

Proposition \ref{prop:ROA} shows that using the CLF candidate \eqref{eq:CLFCandidateFromLQR}, the provable subset of the ROA under the Sontag-type controller is always at least as large as the one under the LQR for the nonlinear system.
Since the LQR stabilizes the linear system, there exists a control input to locally satisfy the CLF inequality which in turn guarantees a local Lyapunov decrease under the Sontag-type controller.
Note that the considered regions in Proposition \ref{prop:ROA} are ellipsoidal since they are defined as sublevel sets of a positive definite quadratic form. Since the proof is based on linearization, the considered regions, however, might be arbitrarily small.\\
The steps to compute the presented control \eqref{eq:sontag_sackmann} with appropriate CLF to achieve local optimality are shown in Algorithm \ref{alg:controller_computation}.
\begin{algorithm}
	\caption{Locally Optimal Controller Computation}\label{alg:controller_computation}
	\begin{algorithmic}[1]
		\State Linearize system \eqref{eq:input_affine_system} around equilibrium $(x_\mathrm{s}=\bm{0}, u_\mathrm{s}=\bm{0})$ to obtain $(A,B)$.
		\State \textbf{Require:} $(A,B)$ stabilizable.
		\State Choose weighting matrices $Q, R \succ 0$ for LQR design.
		\State Compute Riccati matrix $P$ for given $(A,B)$ and $Q,R$. 
		\State Compute feedback \eqref{eq:sontag_sackmann} using the CLF candidate \eqref{eq:CLFCandidateFromLQR}.
	\end{algorithmic}
\end{algorithm}

\section{Global CLF from LQR and Feedback Linearization}
This section shows how the proposed controller design can directly benefit from the large ROA often generated by \emph{feedback linearization} methods.
A globally feedback-linearizable system \eqref{eq:input_affine_system} admits a global diffeomorphism into coordinates $z = T(x)$, $T(\bm{0})=\bm{0}$, where the transformed system reads
\begin{align} \label{eq:feedb_lin_system_dynamics}
	\dot{z} = \tilde{A} z + \tilde{B} \left( \psi(z) + \gamma(z) u \right),
\end{align}
with nonsingular matrix $\gamma(z)$, matrix $\psi(z)$, with $(\tilde{A}, \tilde{B})$ controllable, \cite{Khalil2002NonlinearSystems}.
A known approach to finding a global CLF for the transformed system \eqref{eq:feedb_lin_system_dynamics} is to choose a form that is quadratic in the transformed coordinates $z$ \cite{Sepulchre1997}
\begin{align}
	\label{eq:quadratic_clf_in_z}
	V(z) = \frac{1}{2} z^\top \tilde{P} z.
\end{align}
In the transformed coordinates $z$, the system can be made linear by implementing a preliminary feedback, hence the task is reduced to finding a CLF for an LTI system. When a CLF is found for this linear system, undoing the input and the state transformation thus give a global CLF and a globally stabilizing feedback for the original nonlinear system, which ensures that the CLF condition holds globally in the original coordinates $x$.
The approach in \cite{Sepulchre1997} allows a broad set of matrices in \eqref{eq:quadratic_clf_in_z} as long as they satisfy the CLF condition and does not give hints which choice might be beneficial for performance. An interesting question that arises here is whether we can choose a CLF that is quadratic in the transformed coordinates $z$ (as in \cite{Sepulchre1997}) while also aligning the level sets locally with the CLF $\frac{1}{2} x^\top P x$ from the LQR, as in Proposition \ref{prop:local_optimality}.
This leads to the following (novel) proposition.
\begin{proposition} \label{prop:feedb_lin_quadratic_clf}
	For a feedback linearizable system \eqref{eq:input_affine_system}, the Sontag-type controller \eqref{eq:sontag_sackmann} leads to local optimality as in Proposition \ref{prop:local_optimality} and to global stability given the CLF
	\begin{align} \label{eq:feedb_lin_global_clf}
		V(x) = \frac{1}{2} z^\top(x) \tilde{P} z(x),
	\end{align}
	is used, with
	\begin{align} \label{eq:feedb_lin_global_clf_matrix}
		\tilde{P} = \left(\pdiff{T}{x}|_{x=\bm{0}}\right)^{-\top} P \left( \pdiff{T}{x}|_{x=\bm{0}} \right)^{-1}.
	\end{align}
\end{proposition}
\begin{proof}
	Given the CLF \eqref{eq:feedb_lin_global_clf}, the choice of $\tilde{P}$ leading to local optimality for the resulting Sontag-type controller is to align it with \eqref{eq:CLFCandidateFromLQR} locally. Replacing $z$ by its Taylor series expansion in \eqref{eq:quadratic_clf_in_z} and demanding equality of \eqref{eq:quadratic_clf_in_z} up to the quadratic parts to \eqref{eq:CLFCandidateFromLQR} leads to
	\begin{align}
		\label{eq:aligning_clfs}
		\begin{split}
		V(x) &= \frac{1}{2} x^\top \left( \pdiff{T}{x}|_{x=\bm{0}} \right)^\top \tilde{P} \left( \pdiff{T}{x}|_{x=\bm{0}} \right) x  + R_3(x)\\ &\overset{!}{\approx} \frac{1}{2} x^\top P x.
		\end{split}
	\end{align}
	The obvious choice for $\tilde{P}$ aligning the quadratic parts in \eqref{eq:aligning_clfs} is \eqref{eq:feedb_lin_global_clf_matrix}.
	Since $P \succ 0$ and $T$ is a diffeomorphism, $\tilde{P} \succ 0$ holds.
	
	Proposition \ref{prop:ROA} ensures that there exists a compact sublevel set $\mathbb{X}_{\mathrm{ROA,L}}$ of $V$ as in \eqref{eq:X_ROA_LQR} where the CLF condition \eqref{eq:CLF_condition} holds. 
	Applying the diffeomorphism to transform into feedback-linearized coordinates $z = T(x)$ and using the input transformation
	\begin{align} \label{eq:feedb_lin_input_transf}
		u = \gamma^{-1}(z) \left( -\psi(z) + v \right)
	\end{align}
	leads to a transformed CLF condition
	\begin{align} \label{eq:feedb_lin_inf_in_z}
		\inf_{v, z \in \mathbb{Z}_{\mathrm{ROA,L}}} \left\{ \pdiff{V}{z} \left(\tilde{A}z + \tilde{B} v \right) \right\} < 0,
	\end{align}
	fulfilled in the transformed set
	\begin{align}
		\mathbb{Z}_{\mathrm{ROA,L}} = \left\{ T(x) | x \in \mathbb{X}_{\mathrm{ROA,L}} \right\}.
	\end{align}
	Expanding the terms in \eqref{eq:feedb_lin_inf_in_z} leads to the following inequality:
	\begin{align} \label{eq:feedb_lin_clf_cond_local}
		\inf_{v, z \in \mathbb{Z}_{ROA,L}} \left\{ \underbrace{\frac{1}{2} \trans{z} \left( \trans{\tilde{A}} \tilde{P} + \tilde{P} \tilde{A} \right) z}_{:= \alpha(z)} + \underbrace{\trans{z} \tilde{P} \tilde{B}}_{:= \beta(z)} v \right\} < 0.
	\end{align}
	Since \eqref{eq:feedb_lin_clf_cond_local} holds in an open region around the origin, a simple case distinction whether or not $\alpha(z)$ is negative definite shows that \eqref{eq:feedb_lin_clf_cond_local} necessarily also holds globally for all $z$. Consider the non-trivial case where $\alpha(z)$ is non-negative definite. 
	Since the whole expression \eqref{eq:feedb_lin_clf_cond_local} holds in an open region around the origin, all the directions in which $\alpha(z)$ is not smaller than $0$, $\beta(z) \neq 0$ holds necessarily. Since $\alpha(z)$ is a quadratic form, and $\beta(z)$ is linear, this also holds globally.
\end{proof}

Note that the transformed coordinates $z$ are only used in the construction of the CLF \eqref{eq:quadratic_clf_in_z}. The control design \eqref{eq:sontag_sackmann} is still done in the original coordinates $x$.
Compared to the related CLF design in \cite{Sepulchre1997}, constructing the CLF in Proposition \ref{prop:feedb_lin_quadratic_clf} can be done without knowledge of the transformed system \eqref{eq:feedb_lin_system_dynamics} and a solution to a matrix-valued inequality.

Proposition \ref{prop:feedb_lin_quadratic_clf} assumes that the feedback linearization can be done globally. For many systems this property holds only locally meaning that the resulting CLF will be limited to the same region.

\section{Numerical Results} \label{sec:numerical_results}
To validate the proposed control design, we perform numerical simulations on the well-known inverted pendulum system \cite{Khalil2002NonlinearSystems}. The state vector consists of the angle deviation $\theta$ from the upper equilibrium and its derivative. The input is the cart's acceleration $u = \ddot{y}$.
\begin{align}
	\label{eq:inverted_pendulum_ode}
	\begin{bmatrix}
		\dot{\theta}\\
		\ddot{\theta}
	\end{bmatrix}
	=
	\begin{bmatrix}
		\dot{\theta}\\
		\frac{mgL \sin\left(\theta\right)}{(J+mL^2)}
	\end{bmatrix}
	+
	\begin{bmatrix}
		0\\
		\frac{- mL \cos\left(\theta\right)}{(J+mL^2)}
	\end{bmatrix}
	u
\end{align}
All simulations are performed using a fixed-step 4th order Runge-Kutta integration scheme with a step size of $\SI{0.01}{\s}$.

The code used generate the controllers and simulation results is publicly accessible \footnote{\url{https://github.com/TUMRT/locally_optimal_sontag_type_control}}.


For a feedback linearizable system \eqref{eq:input_affine_system}, four designs can be made for achieving local optimality:
\begin{itemize}
	\item Design (i): Sontag-type controller with local CLF from LQR design (see Algorithm \ref{alg:controller_computation}).
	\item Design (ii): Sontag-type controller as in Design (i) but replace the CLF candidate \eqref{eq:CLFCandidateFromLQR} in Step $5$ in Algorithm \ref{alg:controller_computation} by \eqref{eq:feedb_lin_global_clf}, as in Proposition \ref{prop:feedb_lin_quadratic_clf}.
	\item Design (iii): Design a controller purely based on feedback linearization \eqref{eq:feedb_lin_input_transf} and choose the linear in $z$ feedback $v(x) = - K z(x)$ s.t. $\pdiff{u}{x}|_{x=\bm{0}} = -\inv{R}\trans{B}P$, to ensure local optimality.
	\item Design (iv): LQR for the linearized system around the equilibrium.
\end{itemize}
All designs assume the same weighting matrices $Q,R \succ 0$ and all require the corresponding Riccati matrix $P$ from the LQR design at the equilibrium.

Note that the inverted pendulum system \eqref{eq:inverted_pendulum_ode} is already of the form \eqref{eq:feedb_lin_system_dynamics}, and therefore Design (i) and Design (ii) are identical.
However, since the assumption $\gamma(x)$ invertible in Proposition \ref{prop:feedb_lin_quadratic_clf} does not hold at $\theta = \pm \frac{\pi}{2}$, the results of Proposition \ref{prop:feedb_lin_quadratic_clf} are also valid only when the initial angle satisfies $\theta_0 \in \left(-\frac{\pi}{2}, \frac{\pi}{2}\right)$.

For the presented example, all designs are done using the same weighting matrices $Q=I, \, R=1$.
Design (iii) is another approach giving local optimality and large ROA for feedback linearizable systems and is only used for comparison.

\begin{figure}[h]
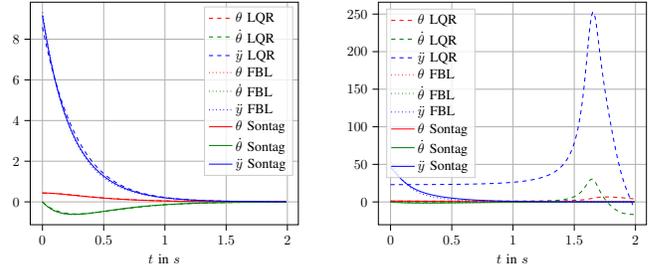

\centering
\begin{subfigure}{0.25\textwidth}
	\scalebox{0.52}{\input{fig/simple_sim_res_25deg.tex}}
\end{subfigure}%
\begin{subfigure}{0.25\textwidth}
	\scalebox{0.52}{\input{fig/simple_sim_res_67deg.tex}}
\end{subfigure}
\caption{Closed loop responses of the inverted pendulum \eqref{eq:inverted_pendulum_ode} with LQR (Design (iv)), locally optimal feedback linearization-based (FBL, Design (iii)) and Sontag-type controller (Design (i)) respectively to different initial angles.}
\label{fig:cl_behavior}
\end{figure}
Figure \ref{fig:cl_behavior} shows states and input trajectories of the closed-loop system for different initial conditions $x_0 = \begin{bmatrix}
	\theta_0, 0
\end{bmatrix}^\top$ with $\theta_0 \in \{25^{\circ}, 67^{\circ}\}$. 
While for the smaller initial angle displacement, the behavior is very similar for all controllers, the LQR fails to stabilize the upper equilibrium for the larger initial angle displacement while the nonlinear controllers are still able to stabilize it.

\begin{figure}[h]
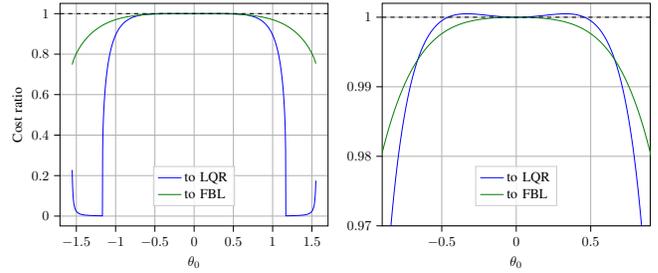

	\centering
	\begin{subfigure}{0.25\textwidth}
		\scalebox{0.52}{\input{fig/simple_PerformanceComparison1D_ratio.tex}}
	\end{subfigure}%
	\begin{subfigure}{0.25\textwidth}
		\scalebox{0.52}{\input{fig/simple_PerformanceComparison1D_ratio_zoom.tex}}
	\end{subfigure}
	\caption{Quadratic performance index ratio of the Sontag-type controller to the LQR (blue) and the controller based on feedback linearization (FBL, green). The plot on the right shows a zoomed in version of the same curves.}
	\label{fig:performance_comparison_ratio}
\end{figure}

A quantitative performance comparison is shown in Figure \ref{fig:performance_comparison_ratio}, where the discrete-time quadratic performance index
\begin{align}
	J(\theta_0) = \frac{h}{2} \sum_{k=0}^{1499} \trans{x_k} Q x_k + \trans{u_k} R u_k
\end{align}
is shown for simulations over the initial angle displacement $\theta_0$. This represents a zero-order hold discretization of the standard quadratic cost \eqref{eq:standard_optimal_cost}.
For the comparison, a total simulation time of \SI{15}{\s} each is used for all simulations.\\
A complete simulation is performed for both the LQR (Design (iv)), the feedback linearization-based controller (Design (iii)) and the Sontag-type controller with a quadratic CLF (Design (i)/(ii)) as inputs $u_k$, are performed for $1000$ equidistant initial angles $\theta_0 \in \{0^{\circ}, 89^{\circ}\}$. The achieved closed-loop cost for the Sontag-type controller $J_{\mathrm{S}}(\theta_0)$ is divided by the cost achieved by the LQR and the feedback linearizing controller, giving the \emph{Cost ratio to LQR} and \emph{Cost ratio to FBL} respectively as shown in Figure \ref{fig:performance_comparison_ratio}. This represents a relative performance index over the initial angle $\theta_0$ with the Sontag-type controller achieving a better cost if this relative performance index is smaller than $1$.\\
It can be seen that for small initial angle displacements the LQR yields a very similar but sometimes marginally better performance. The Sontag-type controller gives a noticeably smaller cost for larger initial angles as the LQR fails to stabilize the equilibrium once the initial angle displacement becomes larger, indicating a larger ROA for the Sontag-type controller.
Towards the margins of the plot, the Sontag-type controller's cost rises as it also approaches the edge of its ROA. The reason is that the CLF condition fails to hold at $\theta = \pm \frac{\pi}{2}$.
Compared to the locally optimal controller based on feedback linearization (FBL, Design (iii)), the Sontag-type controller consistently gives better performance since it does not necessarily compensate beneficial nonlinearities. It is worth noting that there is no straightforward way of further improving the performance of the feedback-linearization-based controller. For the Sontag-type controller, a different CLF more closely approximating the value function might further improve performance (cf. the appendix).

\section{Discussion}

The presented design of Sontag-type controller with quadratic LQR-based CLF is widely applicable. It is able to stabilize equilibria of nonlinear multiple-input systems given the linearization about them is stabilizable and a suitable quadratic cost is chosen. Under these conditions, the feedback is locally optimal i.e. it recovers the LQR at the equilibrium. We have shown that the provable local subset of the ROA of this Sontag-type controller is larger or equal than the one under the LQR serving as the basis of this design. Since the local optimality is based on the replication of the quadratic part of the value function, it seems plausible that another approximation of the value function including higher-order terms but keeping the quadratic part as discussed might improve performance away from the equilibrium or at least enlarge the ROA while keeping local optimality.
For feedback linearizable systems, we have shown how to choose a CLF that leads to a globally stable controller that exhibits the same local optimality as the design above. This design combines these desirable properties which are usually sought using mutually exclusive control designs.
The presented easy way to recover local optimality with the Sontag-type controller uses the LQR design to give a CLF matching the quadratic part of the value function. This is less straightforward for non-quadratic costs.




\section*{APPENDIX}


This section shows the close relation that the Sontag-type controller \eqref{eq:sontag_sackmann} has to optimal control via the choice of CLF. The control law \eqref{eq:sontag_sackmann} is shown to yield the optimal controller, i.e. one satisfying the Hamilton-Jacobi-Bellman equation (HJBE) \cite{Lewis2012} given the value function is used as a CLF during the design, i.e. $V(x) = V^*(x)$. Denote $a^*(x) = \pdiff{V^*}{x}f(x)$, $b^*(x) = \pdiff{V^*}{x}G(x)$. Assume that there exists an optimal feedback $u^*(x)$ satisfying the HJBE, which therefore takes the form \cite{Sepulchre1997}
\begin{align} \label{eq:optimal_feedback}
	u^*(x) = -\inv{R} \trans{{b^*}}(x).
\end{align}
Substituting the optimal control law \eqref{eq:optimal_feedback} into the HJBE gives
\begin{align} \label{eq:optimal_clf_a}
	0 = \frac{1}{2}\trans{x}Qx + a^*(x)  - \frac{1}{2} b^*(x) \inv{R} \trans{{b^*}}(x).
\end{align}

Rearranging \eqref{eq:optimal_clf_a} for $a^*(x)$ and substituting $a^*(x)$ and $b^*(x)$ into the Sontag-type control law \eqref{eq:sontag_sackmann} yields a distortion factor $\lambda(x) \equiv 1$. This shows that the Sontag-type controller is identical to the optimal feedback \eqref{eq:optimal_feedback} in case the value function is used as CLF \cite{Lohmann2023DiscussionOnNonlinQCAndSontagsFormula}.

\printbibliography

@book{Khalil2002NonlinearSystems,
   author = {Hassan K. Khalil},
   edition = {Third Edition},
   isbn = {9780130673893},
   publisher = {Prentice Hall},
   title = {Nonlinear Systems},
   url = {https://books.google.de/books?id=t\_d1QgAACAAJ},
   year = {2002}
}

@book{Freeman1996RobustNonlinearControlDesign,
   author = {Randy A. Freeman and Petar Kokotović},
   city = {Boston, MA},
   doi = {10.1007/978-0-8176-4759-9},
   isbn = {978-0-8176-4758-2},
   publisher = {Birkhäuser Boston},
   title = {Robust Nonlinear Control Design},
   url = {http://link.springer.com/10.1007/978-0-8176-4759-9},
   year = {1996}
}

@inproceedings{Freeman1996CLFNewIdeasFromAnOldSource,
   author = {R.A. Freeman and J.A. Primbs},
   doi = {10.1109/CDC.1996.577294},
   isbn = {0-7803-3590-2},
   booktitle = {Proceedings of 35th IEEE Conference on Decision and Control},
   pages = {3926-3931},
   publisher = {IEEE},
   title = {Control Lyapunov functions: new ideas from an old source},
   volume = {4},
   url = {http://ieeexplore.ieee.org/document/577294/},
   year = {1996}
}

@article{Sontag1989UniversalConstructionOfArtsteinsTheorem,
   author = {Eduardo D. Sontag},
   doi = {10.1016/0167-6911(89)90028-5},
   issn = {01676911},
   issue = {2},
   journal = {Systems and Control Letters},
   pages = {117-123},
   title = {A 'universal' construction of Artstein's theorem on nonlinear stabilization},
   volume = {13},
   year = {1989}
}

@book{Anderson1989,
   author = {B.D.O. Anderson and J.B. Moore},
   isbn = {0-13 -638651-2},
   publisher = {Prentice-Hall International, Inc.},
   title = {Optimal Control},
   year = {1989}
}

@book{Sepulchre1997,
   author = {R. Sepulchre and M. Janković and P. V. Kokotović},
   city = {London},
   doi = {10.1007/978-1-4471-0967-9},
   isbn = {978-1-4471-1245-7},
   publisher = {Springer London},
   title = {Constructive Nonlinear Control},
   url = {http://link.springer.com/10.1007/978-1-4471-0967-9},
   year = {1997}
}

@phdthesis{Primbs1999Diss,
   author = {J.A. Primbs},
   institution = {California Institute of Technology},
   pages = {129},
   title = {Nonlinear Optimal Control: A Receding Horizon Approach},
   url = {http://www.stanford.edu/~japrimbs/Publications/primbsThesis - Copy.pdf},
   year = {1999}
}

@article{Sackmann2000ModifiedOptimalControl,
   author = {Martin S. Sackmann and Volker G. Krebs},
   doi = {10.1016/s1474-6670(17)37190-2},
   issn = {14746670},
   issue = {13},
   journal = {IFAC Proceedings Volumes},
   pages = {199-204},
   publisher = {Elsevier},
   title = {Modified Optimal Control: Global Asymptotic Stabilization of Nonlinear Systems},
   volume = {33},
   url = {http://dx.doi.org/10.1016/S1474-6670(17)37190-2},
   year = {2000}
}

@article{Primbs1998CLFBasedRHC,
   author = {James A. Primbs and Monica Giannelli},
   doi = {10.1109/cdc.1998.758478},
   isbn = {0780343948},
   issn = {01912216},
   issue = {December},
   journal = {Proceedings of the IEEE Conference on Decision and Control},
   pages = {1382-1383},
   title = {Control Lyapunov function based receding horizon control for time-varying systems},
   volume = {2},
   year = {1998}
}

@article{Lohmann2023DiscussionOnNonlinQCAndSontagsFormula,
   author = {Boris Lohmann and Joscha Bongard},
   month = {1},
   title = {A Discussion on Nonlinear Quadratic Control and Sontag's Formula},
   url = {http://arxiv.org/abs/2301.10594},
   year = {2023}
}

@article{Yu2001ComparisonOfNonlinearControl,
   author = {Jie Yu and Ali Jadbabaie and James Primbs and Yun Huang},
   doi = {10.1016/S0005-1098(01)00149-2},
   issn = {00051098},
   issue = {12},
   journal = {Automatica},
   pages = {1971-1978},
   title = {Comparison of nonlinear control design techniques on a model of the Caltech ducted fan},
   volume = {37},
   year = {2001}
}

@book{Lewis2012,
   author = {Frank L. Lewis and Draguna L. Vrabie and Vassilis L. Syrmos},
   doi = {10.1002/9781118122631},
   isbn = {9780470633496},
   journal = {Optimal Control: Third Edition},
   month = {1},
   pages = {1-540},
   publisher = {Wiley},
   title = {Optimal Control},
   url = {https://onlinelibrary.wiley.com/doi/book/10.1002/9781118122631},
   year = {2012}
}

@article{Manchester2017,
   author = {Ian R. Manchester and Jean Jacques E. Slotine},
   doi = {10.1109/TAC.2017.2668380},
   issn = {00189286},
   issue = {6},
   journal = {IEEE Transactions on Automatic Control},
   pages = {3046-3053},
   title = {Control Contraction Metrics: Convex and Intrinsic Criteria for Nonlinear Feedback Design},
   volume = {62},
   year = {2017}
}

\end{document}